\documentclass[11pt]{article}
\usepackage[margin=1in,bottom=1in]{geometry}
\usepackage[utf8]{inputenc}
\usepackage[english]{babel}
\usepackage[all,defaultlines=3]{nowidow}
\usepackage[mathlines]{lineno}
\usepackage{amsmath}
\usepackage{amsthm}
\usepackage{amssymb}
\usepackage{mathtools}
\usepackage{tikz-cd}
\usepackage[T1]{fontenc}
\usepackage{wasysym}
\usepackage{dsfont}
\usepackage[many]{tcolorbox}
\usepackage{soul}
\usepackage{xcolor}
\usepackage{appendix}
\usepackage{enumerate}
\usepackage{multicol}
\usepackage{fancyhdr}
\usepackage{parskip}
\usepackage{hyperref}
\usepackage{mathtools}
\usepackage{nameref}
\usepackage{thmtools}
\usepackage{thm-restate}
\usepackage{adjustbox}
\usepackage{cleveref}
\usepackage{tikz}
\usepackage{float}
\usepackage{MnSymbol}
\usepackage{amsthm}
\usepackage{orcidlink}
\usepackage{ifthen}
\usetikzlibrary{graphs.standard} 
\usepackage{tkz-euclide}

\hypersetup{
  colorlinks   = true, 
  urlcolor     = {blue!50!black}, 
  linkcolor    = {blue!50!black}, 
  citecolor   = {red!50!black} 
}

\newcommand{\SC}{\mathcal{SC}}

\newcommand{\dist}{\mathrm{dist}}

\newcommand{\C}{\mathcal{C}}

\newtheorem{theorem}{Theorem}[section]
\newtheorem*{theorem*}{Theorem}
\newtheorem{lemma}[theorem]{Lemma}
\newtheorem{observation}[theorem]{Observation}
\newtheorem{proposition}[theorem]{Proposition}
\newtheorem*{proposition*}{Proposition}
\newtheorem{corollary}[theorem]{Corollary}
\newtheorem*{corollary*}{Corollary}
\newtheorem*{fact*}{Fact}

\theoremstyle{definition}
\newtheorem*{definition*}{Definition}
\newtheorem{definition}[theorem]{Definition}

\newtheorem*{question*}{Question}
\newtheorem{example}[theorem]{Example}

\newtheorem*{notation*}{Notation}

\theoremstyle{remark}

\newcommand{\Ind}[1]
{#1\setbox0=\hbox{$#1x$}\kern\wd0\hbox to 0pt{\hss$#1\mid$\hss} \lower.9\ht0\hbox to 0pt{\hss$#1\smile$\hss}\kern\wd0}

\newcommand{\notind}[1]
{#1\setbox0=\hbox{$#1x$}\kern\wd0
\hbox to 0pt{\mathchardef\nn=12854\hss$#1\nn$\kern1.4\wd0\hss}
\hbox to 0pt{\hss$#1\mid$\hss}\lower.9\ht0 \hbox to 0pt{\hss$#1\smile$\hss}\kern\wd0}

\newcommand{\N}{\mathbb{N}}

\newcommand{\Gaif}{\mathsf{Gaif}}

\newcommand{\Dcal}{\ensuremath{\mathcal{D}}}

\newcommand{\Hcal}{\ensuremath{\mathcal{H}}}
\newcommand{\Ical}{\ensuremath{\mathcal{I}}}

\newcommand{\Tcal}{\ensuremath{\mathcal{T}}}

\newcommand{\FO}{\ensuremath{\mathsf{FO}}}
\newcommand{\MSO}{\ensuremath{\mathsf{MSO}}}

\title{Extension preservation on dense graph classes}
\date{}

\author{Ioannis Eleftheriadis\footnote{Funded by a George and Marie Vergottis Scholarship awarded by Cambridge Trust, an Onassis Foundation Scholarship, and a Robert Sansom studentship.} \\
    \small{University of Cambridge} \\
     \small{ie257@cam.ac.uk} 
}

\begin{document}

\maketitle

\begin{abstract}
    Preservation theorems provide a direct correspondence between the syntactic structure of first-order sentences and the closure properties of their respective classes of models. A line of work has explored preservation theorems relativised to combinatorially tame classes of sparse structures \cite{extensionspreservation, atseriaspreservation, dawar2010homomorphism, dawar2024preservation}.  In this article we initiate the study of preservation theorems for dense graph classes. In contrast to the sparse setting, we show that extension preservation fails on most natural dense classes of low complexity. Nonetheless, we isolate a technical condition which is sufficient for extension preservation to hold, providing a dense analogue to a result of \cite{extensionspreservation}.

\end{abstract}

\section{Introduction}

The early days of finite model theory were considerably guided by attempts aiming to relativise theorems and techniques of classical model theory to the finite realm. While many of these were trivially shown to admit no meaningful relativisation, others were extended in a way that broadened their applicability and rendered them extremely useful tools in the study of finite models. Preservation theorems were at the heart of this approach. 
Most notably, the {\L}o{\'s}-Tarski preservation theorem which asserts that a first-order formula is preserved by embeddings between all structures if and only if it is equivalent to an existential formula, was shown to fail in the finite from early on \cite{Tait_1959, gurevich1984toward}.  On the contrary, the homomorphism preservation theorem asserting that a formula is preserved by homomorphisms if and only if it is existential-positive, was open for several years until it was surprisingly shown to extend to finite structures \cite{rossman2008homomorphism}, leading to applications in constraint satisfaction problems and database theory. 

Still, considering all finite structures allows for combinatorial complexity, giving rise to wildness from a model-theoretic perspective, and intractability from a computational perspective. Indeed, problems which are hard in general become tractable when restring to classes of finite structures which are, broadly-speaking, tame \cite{dawar2007finite}. In the context of preservation theorems, restricting on a subclass weakens both the hypothesis and the conclusion, therefore leading to an entirely new question. A study of preservation properties for such restricted classes of finite structures was initiated in \cite{atseriaspreservation} and \cite{extensionspreservation} for homomorphism and extension preservation respectively. This investigation led to the introduction of different notions of wideness, which allow for arguments based on the \emph{locality} of first-order logic. However, as it was recently realised \cite{dawar2024preservation}, these arguments require slightly restrictive closure assumptions which are not always naturally present. In particular, it was shown that homomorphism preservation holds over any hereditary \emph{quasi-wide} class which is closed under \emph{amalgamation over bottlenecks}. 

Quasi-wideness is a Ramsey-theoretic condition which informally says that in any large enough structure in the class one can remove a bounded number of elements, called bottleneck points, so that there remains a large set of pairwise far-away elements. Here, the number of bottleneck points is allowed to dependent on the choice of distance. Hereditary quasi-wide classes were later identified with \emph{nowhere dense} classes \cite{firstorderproperties}. Over the years, a successful program was developed aiming to understand the combinatorial and model-theoretic features of nowhere dense classes, and exploit them for algorithmic purposes \cite{nevsetvril2012sparsity}. The culmination of this was the seminal result that first-order model checking is fixed-parameter tractable on nowhere dense classes \cite{deciding}.

In recent years, the focus has shifted towards extending this well understood theory to more general, possibly dense, well-behaved classes, which fall out of the classification provided by the sparsity program. In these efforts, the model-theoretic notions of \emph{monadic stability} and \emph{monadic dependence} have played central roles. Monadic stability, initially introduced by Baldwin and Shelah \cite{baldwin1985second} in the context of classification of complete first-order theories, prohibits arbitrarily large definable orders in monadic expansions. In the language of first-order transductions, a class is monadically stable whenever it does not transduce the class of finite linear orders. More generally, a class is monadically dependent if it does not transduce the class of all graphs. In the context of monotone classes of graphs, Adler and Adler \cite{adler} first observed that the above notions coincide with nowhere density, a result which was also extended to arbitrary relational structures \cite{monoNIP}. The generalisation of sparsity theory to dense classes eventually led to the result that first-order model checking is fixed-parameter tractable on all monadically stable graph classes \cite{dreier2023first}, which in particular include transductions of nowhere dense classes. It is conjectured that the above result extends to all monadically dependent classes, while a converse was recently established for hereditary graph classes (under standard complexity-theoretic assumptions) \cite{flipbreak}. 

The purpose of the present article is to initiate the investigation of preservation theorems on tame dense graph classes. Much like nowhere dense classes are equivalently characterised by quasi-wideness, monadically stable and monadically dependent graph classes also admit analogous wideness-type characterisations. In the case of monadic stability, the relevant condition is known as \emph{flip-flatness} \cite{indiscernibles}; this may be viewed as a direct analogue of quasi-wideness which replaces the vertex deletion operation by flips, i.e. edge-complementations between subsets of the vertex set. For monadically dependent classes the relevant condition, known as \emph{flip-breakability} \cite{flipbreak}, allows to find two large sets such that elements in one are far away from elements in the other, again after performing a bounded number of flips. However, unlike quasi-wideness which was introduced in the context of preservation and then shown to coincide with nowhere density, these conditions were introduced purely for the purpose of providing combinatorial characterisations of monadic stability and monadic dependence respectively. The immediate question thus becomes whether these conditions, or variants thereof, can be used to obtain preservation in restricted tame dense classes, in analogy to the use of wideness in \cite{extensionspreservation,atseriaspreservation,dawar2010homomorphism,dawar2024preservation}. 

The first observation is that the arguments for homomorphism preservation are not directly adaptable in this context due to the nature of flips. Indeed, while the vertex-deletion operation respects the existence of a homomorphism between two structures, the flip operation is not at all rigid with respect to homomorphisms precisely because the latter do not reflect relations, e.g. the graph $K_1 + K_1$ homomorphically maps to $K_2$, but $\overline{(K_1 + K_1)}=K_2$ does not map to $\overline{K_2}=K_1 + K_1$. This issue evidently disappears if one considers embeddings. As it was observed in \cite[Corollary 2.3]{dawar2024preservation}, the extension preservation property implies the homomorphism preservation property in hereditary classes of finite structures, so considering extension preservation is more general for our purposes. 

However, this generality comes at a cost. Indeed, the argument for extension preservation from \cite{extensionspreservation} requires that the number of vertex-deletions is independent of the choice of radius, a condition known as \emph{almost-wideness}. This is a more restrictive assumption which therefore applies to fewer sparse classes. It is not known whether extension preservation is obtainable for quasi-wide classes. At the same time, unlike \cite[Theorem 4.2]{dawar2024preservation} whose proof is essentially a direct application of Gaifman's locality theorem based on an argument of Ajtai and Gurevich \cite{ajtai}, the proof of extension preservation \cite[Theorem 4.3]{extensionspreservation} is admittedly much more cumbersome. One explanation for this is that the homomorphism preservation argument relies on the fact that the disjoint union operation endows the category of graphs and homomorphisms with \emph{coproducts}, i.e. for any graphs $A,B,C$ if there are homomorphisms $f:A \to C$ and $g:B \to C$ then there is a homomorphism $f+g:A + B \to C$ whose pre-compositions with the respective inclusion homomorphisms $\iota_A:A \to A+B$ and $\iota_B: B \to A+B$ are equal to $f$ and $g$ respectively. On the other hand, no construction satisfies the above property in the category of graphs with embeddings; in fact coproducts do not even exist in the category of graphs with strong homomorphisms (see \cite[Corollary 4.3.15]{KnauerKnauer+2019}). 

Our first contribution is negative, showing that extension preservation can fail on tame dense classes of low complexity. In particular, we show that extension preservation fails on the class of all graphs of (linear) cliquewidth at most $k$, for all $k \geq 4$. This answers negatively a question of \cite{dawar2024preservation}. This is contrary to the sparse picture, where it was shown that extension preservation holds in the class of graphs of treewidth at most $k$, for every $k \in \N$ \cite[Theorem 5.2]{extensionspreservation}. Interestingly, extension preservation holds for the class of all graphs of cliquewidth $2$ as this class coincides with the class of cographs which is known to be well-quasi-ordered \cite{damaschke1990induced}. Our construction is based on the encoding of linear orders via the neighbourhoods of certain vertices. Orders are also central to the original counterexample for the failure of extension preservation in the finite due to Tait \cite{Tait_1959}. There, the orders are crucially presented over a signature with two relation symbols and one constant, which does not allow for a direct translation to undirected graphs. Sadly, the fact that orders appear to provide counterexamples rules out the possibility of using an argument based on flip-breakability to establish preservation. 

The second contribution of the article is positive. In particular, we provide a dense analogue to \cite[Theorem 4.3]{extensionspreservation}. For this, we introduce \emph{strongly flip-flat} classes, i.e. those flip-flat classes such that the number of flips is independent of the choice of radius. Moreover, we formulate the dense analogue of the amalgamation construction, which we call the \emph{flip-sum}, whose existence in the class is necessary for the argument to be carried out. The main theorem (\Cref{thm:main} below) may thus be formulated as saying that extension preservation holds over any hereditary strongly flip-flat class which is closed under flip-sums over bottleneck partitions. 

\section{Preliminaries}\label{sec:prelims}

We assume familiarity with the standard notions of finite model theory and structural graph theory, 
and refer to \cite{ebbinghaus} and \cite{nevsetvril2012sparsity} 
for reference. In this article, graphs shall always refer to simple undirected graphs i.e. structures over the signature $\tau_E = \{E\}$ where $E$ is interpreted as a symmetric and anti-reflexive binary relation. For a graph $G$ we write $V(G)$ for its domain (or vertex set), and $E(G)$ for its edge set. In general, for a $\tau$-structure $A$ and a relation symbol $R \in \tau$ of arity $r \in \N$ we write $R^A\subseteq A^r$ for the interpretation of $R$ in $A$. We shall abuse notation and not distinguish between structures and their respective domains.

Given two structures $A,B$ in the same relational signature $\tau$, a homomorphism $f:A \to B$ is a map that preserves all relations, i.e. for all $R \in \tau$ and tuples $\bar a$ from $A$ we have $\bar a \in R^A \implies f(\bar a) \in R^B$. A \emph{strong} homomorphism is a homomorphism $f:A \to B$ that additionally reflects all relations, i.e. $f(\bar a) \in R^B \implies \bar a \in R^A$. An injective strong homomorphism is called an \emph{embedding} or \emph{extension}. 

A $\tau$-structure $B$ is said to be a \emph{weak substructure} 
of a $\tau$-structure $A$ if $B \subseteq A$ and the inclusion map $\iota: B \hookrightarrow A$ is a homomorphism. Likewise, $B$ is an \emph{induced substructure} of $A$ if the inclusion map is an embedding; we write $B \leq A$ for this. Given a structure $A$ and a subset $S \subseteq A$ we write $A[S]$ for the unique induced substructure of $A$ with domain $S$. An induced substructure $B$ of $A$ is said to be \emph{proper} if $B \subsetneq A$; we write $B \lneq A$ for this. We say that a class of structures in the same signature is \emph{hereditary} if it is closed under induced substructures. 
Moreover a class is called \emph{addable} if it is closed under taking disjoint unions, which we denote by $A+B$. 

By the \emph{Gaifman graph} of a structure $A$ we mean the undirected graph $\Gaif(A)$ with vertex set $A$  such that two elements are adjacent if, and only if, they appear together in some tuple of a relation of $A$. Given a structure $A$, $r \in \N$, and $a \in A$, we write $N^A_r(a)$ for the \emph{$r$-neighbourhood of $a$ in $A$}, that is, the set of elements of $A$ whose distance from $a$ in $\Gaif(A)$ is at most $r$. We shall often abuse notation and write $N^A_r(a)$ for the induced substructure $A[N^A_r(a)]$ of $A$. For a set $C \subseteq A$ we define $N^A_r(C):=\bigcup_{a \in C} N^A_r(a)$.  
A set $S \subseteq A$ is said to be \emph{$r$-independent} if $b \notin N^A_r(a)$ for any $a,b \in S$. 

For $r \in \N$, let $\mathrm{dist}(x,y)\leq r$ be the first-order formula expressing that the distance between $x$ and $y$ in the Gaifman graph is at most $r$, and $\mathrm{dist}(x,y)> r$ its negation. Clearly, the quantifier rank of $\mathrm{dist}(x,y)\leq r$ is at most $r$. A \emph{basic local sentence} is a sentence
\[ \exists x_1, \dots, x_n (\bigwedge_{i \neq j} \mathrm{dist}(x_i,x_j)> 2r \land \bigwedge_{i \in [n]} \psi^{N_r(x_i)}(x_i)),\]
where $\psi^{N_r(x_i)}(x_i)$ denotes the relativisation of $\psi$ to the $r$-neighbourhood of $x_i$, i.e. the formula obtained from $\psi$ by replacing every quantifier $\exists x \ \theta$ with $\exists x (\dist(x_i,x)\leq r \land \theta)$, and likewise every quantifier $\forall x \ \theta$ with $\forall x (\dist(x_i,x)\leq r \to \theta)$. We call $r$ the \emph{locality radius}, $n$ the \emph{width}, and $\psi$ the \emph{local condition} of $\phi$.  Recall the Gaifman locality theorem~\cite[Theorem~2.5.1]{ebbinghaus}.

\begin{theorem}[Gaifman Locality]
    Every first-order sentence of quantifier rank $q$ is equivalent to a Boolean combination of basic local sentences of locality radius $7^q$. 
\end{theorem}

A class $\C$ of structures is said to be \emph{quasi-wide} if for every $r \in \N$ there exist $k_r \in \N$ and $f_r: \N \to \N$ such that for all $m \in \N$ and all $A \in \C$ of size at least $f_r(m)$ there exists $S \subseteq A$ such that $A \setminus S$ contains an $r$-independent set of size $m$. Moreover, if $k_r:=k \in \N$ is independent of $r$, then $\C$ is said to be \emph{almost-wide}. Finally, we say that a class $\C$ is \emph{uniformly} quasi-wide (uniformly almost-wide respectively) if the hereditary closure of $\C$ is quasi-wide (almost-wide respectively). 

For a graph $G$ and a pair of disjoint vertex subsets $U$ and $V$, the subgraph {\em{semi-induced}} by $U$ and $V$ is the bipartite graph with sides $U$ and $V$ that contains all edges of $G$ with one endpoint in $U$ and second in $V$. By the \emph{half-graph of order $n$} we mean the bipartite graph with vertices $\{u_i,v_i: i \in [n]\}$ and edges $\{(u_i,v_j): i\leq j\}$.

Let $c \in \N$ and $\delta$ and $\phi$ be first-order formulas over the signature $\tau_E^c:= \tau_E \cup \{P_1,\dots,P_c\}$. We define the \emph{transduction $T_{\delta,\phi}$} as the operation that maps a graph $G$ to the class $T_{\delta,\phi}(G)$ containing the graphs $H$ such that there exists a $\tau^c_E$-expansion $G^+$ of $G$ satisfying $V(H)=\{v \in V(G): G^+ \models \delta(v)\}$ and 
\[E(H):= \left\{(u,v)\in V(H)^2 \colon u \neq v \land G^+ \models (\phi(u,v) \lor \phi(v,u))\right\}.\]
We write $T_{\delta,\phi}(\C):=\bigcup_{G \in \C} T_{\delta,\phi}(G)$ and say that \emph{$\Dcal$ is a transduction of a class $\C$}, or \emph{$\C$ transduces $\Dcal$}, if there exists a transduction $\Tcal_{\delta,\phi}$ such that  $\Dcal \subseteq I_{\delta,\phi}(\C)$. A graph class $\C$ is \emph{monadically dependent} if $\C$ does not transduce the class of all graphs. $\C$ is moreover \emph{monadically stable} if $\C$ does not transduce the class of all half-graphs.

We say that a formula $\phi$ is preserved by extensions 
over a class of structures $\C$ if for all $A,B \in \C$ such that there is a embedding 
from $A$ to $B$, $A \models \phi$ implies that $B \models \phi$. We say that a class of structures $\C$ has the \emph{extension preservation property} if for every formula $\phi$ preserved by extensions 
over $\C$ there is an existential 
formula $\psi$ such that $M \models \phi \iff M \models \psi$ for all $M \in \C$. We analogously define the \emph{homomorphism preservation property}, replacing ``embeddings'' with ``homomorphisms'' and ``existential'' with ``existential positive'' in the above. 

Given a formula $\phi$ and a class of structures $\C$, we say that $M \in \C$ is a \emph{minimal induced model} of $\phi$ in $\C$ if $M\models \phi$ and for any proper induced substructure $N$ of $M$ with $N \in \C$ we have $N \not\models \phi$. The relationship between minimal models and extensions preservation is highlighted by the following folklore lemma. We provide a proof for completeness.

\begin{lemma}\label{lem:minimalmodels}
    Let $\C$ be a hereditary class of finite structures. Then a sentence preserved by extensions in $\C$ is equivalent to an existential sentence over $\C$ if and only if it has finitely many minimal induced models in $\C$. 
\end{lemma}

\begin{proof}
    Suppose that $\phi$ has finitely many minimal induced models in $\C$, say $M_1,\dots,M_n$. For each $i \in [n]$, let $\psi_i$ be the primitive sentence inducing a copy of $M_i$ and write $\psi:= \bigvee_{i \in [n]} \psi_i$; evidently $\psi$ is existential. We argue that $\phi$ is equivalent to $\psi$ over $\C$. Indeed, if $A \in \C$ models $\phi$ then $A$ contains a minimal induced model $B$ of $\phi$ as an induced substructure. By hereditariness $B \in \C$ and so $B$ is isomorphic to some $M_i$. Since there is clearly an embedding $B \to A$ it follows that $A \models \psi$. On the other hand if $A \models \psi$, then some $M_i$ embeds into $A$. Since $M_i \models \psi$ and $\psi$ is preserved by embeddings this implies that $A \models \phi$ as required.

    Conversely, assume that $\phi$ is equivalent to an existential sentence over $\C$. In particular, $\phi$ is equivalent to some disjunction $\bigvee_{i \in [n]} \psi_i$ where each $\psi_i$ is primitive. It follows that each $\psi_i$ is the formula inducing some structure $M_i$. Now, if $A$ is a minimal induced model of $\phi$ then in particular $A \models \psi_i$ for some $i \in [n]$, i.e. there is an embedding $h:M_i \to A$. If $h$ is not surjective, then $A[h[M_i]]$ is a proper induced substructure of $A$, which is in $\C$ by hereditariness, and models $\phi$; this contradicts the minimality of $A$. Hence, the size of every minimal induced model of $\phi$ in $\C$ is bounded by $\max_{i \in [n]} |M_i|$. It follows that $\phi$ can have only finitely many minimal induced models in $\C$. 
\end{proof}

\section{Failure of  preservation on graphs of cliquewidth 4}\label{sec:ext}

One consequence of \Cref{lem:minimalmodels} is that extension preservation holds over any class $\C$ that is \emph{well-quasi-ordered} by the induced substructure relation, i.e. classes for which there exists no infinite collection of members which pairwise do not embed into one another. In particular, this applies to the class of \emph{cographs} \cite{damaschke1990induced}, which are precisely the graphs of cliquewidth $2$ (see \cite{courcelle} for background on cliquewidth). Hence, one may reasonably inquire whether extension preservation is generally true for the class $\mathsf{Clique}_k$ of all graphs of cliquewidth at most $k$. This would in particular reflect an analogous phenomenon that is true in the sparse setting, that is, that extension preservation holds over the class $\mathsf{Tree}_k$ of all graphs of treewidth at most $k$ \cite[Theorem 5.2]{extensionspreservation}. 

Classes of bounded cliquewidth are not monadically stable, as even the class of cographs contains arbitrarily large semi-induced half-graphs, but they are monadically dependent. In fact, their structural properties imply tame behaviour going much beyond the context of first-order logic (see \cite{dabrowski} for a survey). Still, as it turns out, extension preservation fails even at the level of cliquewidth $4$. To show this, we produce a formula $\phi$ preserved by embeddings over the class of all finite undirected graphs, which admits minimal models of cliquewidth $4$. In particular, \Cref{lem:minimalmodels} implies that extension preservation fails on any class that includes these minimal models. Our idea is based on encoding two interweaving linear orders on the two parts of a semi-induced graph. Two vertices on the same part are comparable in this ordering whenever their neighbourhoods in the other part are set-wise comparable. This effectively forces a semi-induced half-graph. 

Our formula is in the form of an implication, proceeded by a primitive part which induces a gadget corresponding to the beginning and end of the two linear orders. The antecedent of the implication first makes sure that the above relation is a pre-ordering on each side of the semi-induced graph, while it imposes that the vertices of the gadget corresponding the minimal and maximal elements are indeed minimal and maximal in this pre-ordering. Moreover, it essentially makes sure that successors, i.e. vertices of the same part whose neighbourhoods over the other part differ by a single element, are adjacent on one part and non-adjacent on the other. The consequent then imposes that any vertex on the first side has an adjacent successor, while every vertex on the other side has a non-adjacent successor. Because each one of the two pre-orders precisely compares neighbourhoods over the other part, this forces the pre-orders to be anti-symmetric, and thus the two parts to have the same number of vertices. 

Finally, two additional vertices are also added on one side of the semi-induced bipartite graph, which are part of the gadget and serve no role in this ordering. These make sure that our intended minimal models form an anti-chain in the embedding relation, as they crucially result in the existence of a unique embedding of the gadget into the models (\Cref{lem:Iembeds} below).

We now turn to formal definitions. Let $I(v_1,v_2,v_3,v_4,v_5,v_6,u_1,u_2,u_3,u_4,u_5,u_6,a,b)$ be the formula that induces the graph of \Cref{fig:I} below. In the following, we treat the free variables of $I$ as constants for simplicity.  The notation $\forall( x \in U)$ will denote the relativisation of the universal quantifier to the neighbours of $v_1$ that are not $v_2$, i.e. $\forall(x \in U) \ \psi(x)$ is shorthand for $\forall x (E(x,v_1)\land x\neq v_2  \to \psi(x))$. Likewise, the notation $\forall( x \in V)$ denotes the relativisation of the universal quantifier to the non-neighbours of $v_1$ that are not $a$ or $b$, i.e. $\forall(x \in V) \ \psi(x)$ is shorthand for $\forall x (\neg E(x,v_1) \land x\neq a \land x \neq b \to \psi(x))$. Existential quantifiers relativised to $U$ and $V$ are defined analogously. Consider the auxiliary formulas:
\[ x\leq_V y:= \forall (z \in U)[E(z,x) \to E(z,y)];\]
\[ x <_V y:= x\leq_V y \land \neg(y\leq_V x);\]
\[ \chi_1:=\forall (x \in V) \forall (y \in V) [x \leq_V y \lor y \leq_V x];\]
\[ \chi_2:=\forall (x \in U) [ E(x,v_6) \to x=u_6];\]
\[ \chi_3:=\forall (x \in V) \forall (y \in V) [x <_V y \land E(x,y) \to \exists ! (z \in U)(E(y,z) \land \neg E(x,z))].\]
In analogy, we define:
\[ x\leq_U y:= \forall (z \in V)[E(z,x) \to E(z,y)];\]
\[ x <_U y:= x\leq_U y \land \neg(y\leq_U x);\]
\[ \xi_1:=\forall (x \in U) \forall (y \in U) [x \leq_U y \lor y \leq_U x];\]
\[ \xi_2:=\forall (x \in V) [ E(x,u_1) \to x=v_1];\]
\[ \xi_{2^*}:=\forall (x \in V) \ E(x,u_6);\]
\[ \xi_3:=\forall (x \in U) \forall (y \in U) [x <_U y \land \neg E(x,y) \to \exists ! (z \in V)(E(y,z) \land \neg E(x,z))].\]
We then define:
\[ \phi_1:= \chi_1 \land \chi_2 \land \chi_3 ;\]
\[ \phi_2:=\forall (x \in V)[x \neq v_1 
\to \exists (y \in V)(E(x,y) \ \land \ x <_V y)];\]
\[ \psi_1:=\xi_1\land \xi_2 \land \xi_{2^*} \land \xi_3;\]
\[ \psi_2:=\forall (x \in U)[x \neq u_6 \to \exists (y \in U)(\neg E(x,y) \land x<_U y)].\]
Putting the above together we finally define:
\[ \phi:= \exists \bar v,\bar u,a,b ( I(\bar v,\bar u,a,b) \land 
[\phi_1(\bar v,\bar u,a,b) \land \psi_1(\bar v,\bar u,a,b) \to \phi_2(\bar v,\bar u,a,b) \land \psi_2(\bar v,\bar u,a,b) ])\]

\begin{figure}[h!]
    \centering
    \begin{tikzpicture}[scale=1.5, every node/.style={circle, draw, inner sep=1.5pt}]
    \node (u1) at (1,1.5) {$u_1$};
    \node (u2) at (2,1.5) {$u_2$};
    \node (u3) at (3,1.5) {$u_3$};
    \node (u4) at (5,1.5) {$u_4$};
    \node (u5) at (6,1.5) {$u_5$};
    \node (u6) at (7,1.5) {$u_6$};
    \node (v1) at (1,0) {$v_1$};
    \node (v2) at (2,0) {$v_2$};
    \node (v3) at (3,0) {$v_3$}; 
    \node (v4) at (5,0) {$v_4$};
    \node (v5) at (6,0) {$v_5$};
    \node (v6) at (7,0) {$v_6$};
    \node (A) at (2,-1) {$a$};
    \node (B) at (6,-1) {$b$};

    \foreach \i in {1,...,6} {
        \foreach \j in {\i,...,6} {
            \draw (v\i) -- (u\j);
        }
    }
    \foreach \i in {3,...,6} {
    \draw (A) -- (u\i);
    \draw (u1) edge[out=45, in=135] (u\i);
    }
    \foreach \i in {4,...,6} {
    \draw (u2) edge[out=45, in=135] (u\i);
    \draw (u3) edge[out=45, in=135] (u\i);
    }
    \draw (u4) edge[out=45, in=135] (u6);
    \draw (v1) -- (v2);
    \draw (v2) -- (v3);
    \draw (v4) -- (v5);
    \draw (v5) -- (v6);
    \draw (A) edge[out=135, in=215] (u2);
    \draw (A) -- (v2);
    \draw (B) edge[out=135, in=215] (u5);
    \draw (B) -- (v5);
    \draw (B) -- (u6);
    
\end{tikzpicture}
    \caption{The gadget induced by $I(\bar v,\bar u,a,b)$.}
    \label{fig:I}
\end{figure}
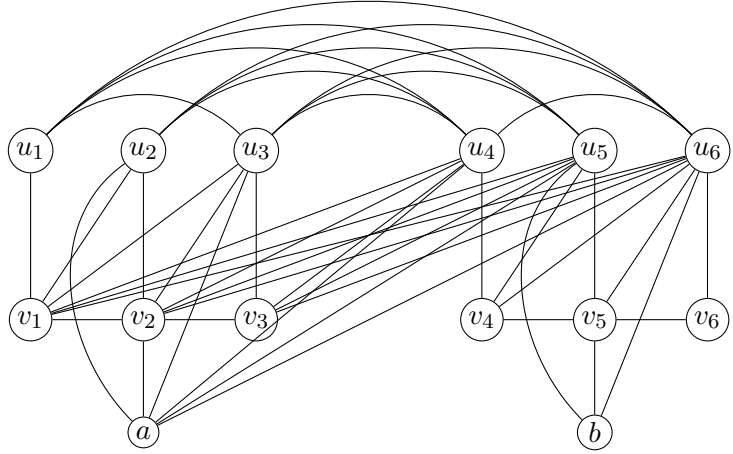

\begin{proposition}
    The formula $\phi$ is preserved by embeddings over the class of all finite graphs. 
\end{proposition}

\begin{proof}
    Let $G,H$ be two graphs such that $G$ embeds into $H$, and $G \models \phi$. Without loss of generality we assume that $V(G) \subseteq V(H)$ and that the identity map is an embedding. We shall argue that $H \models \phi$. 

    Since $G \models \phi$, we may fix (tuples of) vertices $\bar v,\bar u,a,b \in V(G)$ such that $G \models I(\bar v,\bar u,a,b)$. Evidently, $H$ also models $I(\bar v,\bar u,a,b)$. If $H \not\models (\phi_1(\bar v,\bar u,a,b)\land \psi_1(\bar v,\bar u,a,b))$ then $H \models \phi$; we may therefore assume that $H \models (\phi_1(\bar v,\bar u,a,b)\land \psi_1(\bar v,\bar u,a,b))$. Let $U:=N_H(v_1)\setminus \{v_2\} \subseteq V(H)$ be the neighbours of $v_1$ in $H$ that are not $v_2$, and $V:=V(H) \setminus (U \cup \{a,b\})\subseteq V(H)$ be the non-neighbours of $v_1$ that are not $a$ or $b$ (in particular $v_1 \in V$). We call the vertices $x \in V$ \emph{$V$-elements}. For each $V$-element $x$ we write $U_x:=\{y \in U: H \models E(x,y)\}$ for the \emph{$U$-neighbourhood of $x$}. We similarly define \emph{$U$-elements} and \emph{$V$-neighbourhoods}. From each of the conjuncts $\chi_i$ of $\phi_1$ we deduce that in $H$:
    \begin{itemize}
        \item[$\chi_1$:] $V$-elements have pairwise comparable $U$-neighbourhoods;
        \item[$\chi_2$:] the only member of $U_{v_6}$ is $u_6$;
        \item[$\chi_3$:] if two adjacent $V$-elements $x,y$ satisfy $U_x \subsetneq U_y$ then $|U_y|=|U_x|+1$.
    \end{itemize}

    We shall argue that items $(1)$-$(3)$ are still true within $G$, replacing $V$ with $V':=V \cap V(G)$ and $U$ with $U':=U \cap V(G)$, and so $G \models \phi_1(\bar v,\bar u,a,b)$. We write $U'_x:=U_x \cap B$ for the relativised $U'$-neighbourhoods. Clearly, items $(1)$ and $(2)$ are still true in $G$. For item $(3)$, suppose that $x,y \in V'$ are two adjacent $V'$-elements, such that $V'_x \subsetneq V'_y$. Since $V$-elements in $H$ have pairwise comparable $U$-neighbourhoods, we deduce that $V_x \subsetneq V_y$ and therefore that $|V_y|=|V_x|+1$ as $H$ satisfies $\chi_3$. In particular, it follows that $|V'_y|=|V'_x|+1$ as required, and so $G$ models $\chi_3(\bar v,\bar u,a,b)$ and consequently $\phi_1(\bar v,\bar u,a,b)$.

    Similarly, from each of the conjuncts $\xi_i$ of $\psi_1$ we deduce that in $H$:
      \begin{itemize}
        \item[$\xi_1$:] $U$-elements have pairwise comparable $V$-neighbourhoods;
        \item[$\xi_2$:] the only element of $V_{u_1}$ is $v_1$, while $V_{u_6}=V$;
        \item[$\xi_{2^*}$:] $V_{u_6}$ is equal to $V$;
        \item[$\xi_3$:] if two non-adjacent $U$-elements satisfy $V_x \subsetneq V_y$ then $|V_y|=|V_x|+1$.
    \end{itemize}

    Arguing as before, we obtain that $G \models \psi_1(\bar v,\bar u,a,b)$. Since $G \models \phi$ and $G \models (\phi_1(\bar v,\bar u,a,b)\land \psi_1(\bar v,\bar u,a,b))$ we deduce that $G \models (\phi_2(\bar v,\bar u,a,b) \land \psi_2(\bar v,\bar u,a,b))$, i.e. the following are true in $G$:
    \begin{itemize}
        \item[$\phi_2$:] every $V'$-element that is not $v_1$ is adjacent to a $V'$-element of strictly greater $U'$-neighbourhood;
        \item[$\psi_2$:] every $U'$-element that is not $u_6$ is non-adjacent to a $U'$-element of strictly greater $V'$-neighbourhood. 
    \end{itemize}

    We proceed to show that the above implies that $V=V'$ and $U'=U$, and hence $G=H$. In particular, this implies that $H \models \phi$ as claimed. 
    
    Since $G$ is finite and satisfies $\phi_2$ we obtain some $n \in \N$ and a sequence of distinct elements $\alpha_1:=v_6,\alpha_2,\dots,\alpha_n:=v_1$ of $V'$ such that $U'_{\alpha_i}\subsetneq U'_{\alpha_{i+1}}$ and $G \models E(\alpha_i,\alpha_{i+1})$ for all $i \in [n-1]$. In particular, $U_{\alpha_i}\subsetneq U_{\alpha_{i+1}}$ and $H \models E(\alpha_i,\alpha_{i+1})$ for all $i \in [n]$. As $H$ satisfies $\chi_3$ we obtain that $|U_{\alpha_{i+1}}|=|U_{\alpha_i}|+1$ for all $i$. 
    Moreover, since $H$ satisfies $\chi_2$ and every element of $U$ is adjacent to $v_1$, we obtain that $U_{\alpha_1}=\{u_6\}$ and  $U_{\alpha_n}=U$. In particular, we deduce that $n=|U|\leq |V'|$. 
    Symmetrically, we obtain some $k \in \N$ and a sequence of elements $\beta_1:=u_1,\beta_2,\dots,\beta_k:=u_6$ of $U'$ such that $V'_{\beta_i}\subsetneq V'_{\beta_{i+1}}$ and $G \models \neg E(\beta_{i},\beta_{i+1})$ for all $i \in [k-1]$. Hence, 
    $V_{\beta_{i}}\subsetneq V_{\beta_{i+1}}$ and $H \models \neg E(\beta_{i},\beta_{i+1})$ for all $i \in [k-1]$. Once again, since $H$ satisfies $\xi_2$, $\xi_{2^*}$, and $\xi_3$ we obtain that $V_{\beta_1}=\{v_1\}$, $V_{\beta_n}=V$, and $|V_{\beta_{i+1}}|= |V_{\beta_{i}}|+1$. It thus follows that $k=|V|\leq |U'|$. Putting the above together we have that 
    \[ |U|\leq |V'| \leq |V| \leq |U'| \leq |U|. \]

    Consequently, $n=k$ while $V=V'=\{\alpha_1,\dots,\alpha_n\}$ and $U=U'=\{\beta_1,\dots,\beta_n\}$ as needed. 
\end{proof}

We now define the intended minimal induced models of our formula $\phi$.

\begin{definition}
    For $n \geq 7$ we define the graph $\Hcal_n$ with vertex and edge set 
    \[ V(\Hcal_n):= \{v_1,\dots,v_n\} \cup \{u_1,\dots,u_n\} \cup \{a\} \cup \{b\}; \]
    \[ E(\Hcal_n):= \{(v_i,u_j):i\leq j\} \cup \{(v_i,v_j): j=i+1\} \cup \{(u_i,u_j): j \neq i+1\}  \]
    \[ \cup \ \{(a,u_i) : i\geq 2\} \cup \{(b,u_i): i \geq n-1\} \cup \{(a,v_2),(b,v_{n-1})\},  \]
    respectively. We also write $\Ical_n$ for the subgraph of $\Hcal_n$ induced on the set 
    \[ V(\Ical_n):=\{v_1,v_2,v_3,v_{n-2},v_{n-1},v_n,u_1,u_2,u_3,u_{n-2},u_{n-1},u_n,a,b\} \subseteq V(\Hcal_n).\] 
\end{definition}

\begin{figure}[h!]
    \centering
    \begin{tikzpicture}[scale=1.5, every node/.style={circle, draw, inner sep=1.5pt}]
    \foreach \i in {1,2,...,7} {
        \node (u\i) at (\i, 1.5) {};
        \node (v\i) at (\i, 0) {};
    }
    \node (A) at (2,-1) {};
    \node (B) at (6,-1) {};

    \foreach \i in {1,...,7} {
    \edef\next{\number\numexpr \ifnum \i<7 \i+1 \else \i\fi}

    \draw (v\i) -- (v\next);
        \foreach \j in {1,...,\i} {
            \draw (u\i) -- (v\j);
        }
    }

    \foreach \i in {1,...,5} {
        \edef\nextt{\number\numexpr \ifnum \i<6 \i+2\else 5\fi}
    \foreach \k in {\nextt,...,7} {
        \draw (u\i) edge[out=45, in=135] (u\k);
    }
    }
    \foreach \i in {3,...,7} \draw (A) -- (u\i);
    \draw (A) edge[out=135, in=215] (u2);
    \draw (A) -- (v2);
    \draw (B) edge[out=135, in=215] (u6);
    \draw (B) -- (v6);
    \draw (B) -- (u7);tik
\end{tikzpicture}
    \caption{The graph $\Hcal_7$.}
\end{figure}
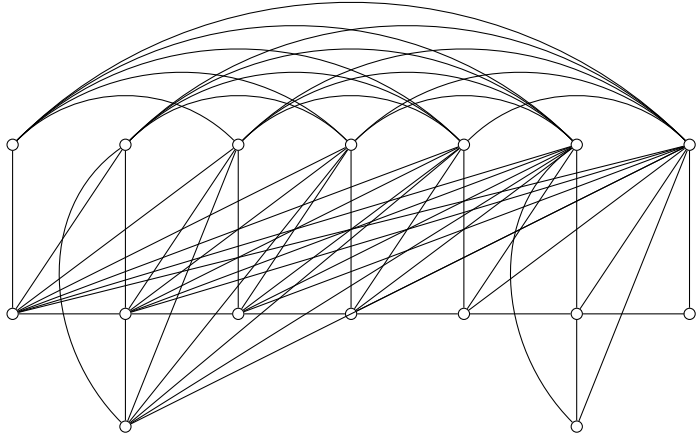

We aim to establish that the graphs $\Hcal_n$ are all minimal induced models of $\phi$. Towards this, we first argue that the only embedding of $\Ical_n$ in $\Hcal_n$ is the inclusion map. While this lemma is not conceptually difficult, it requires analysing and ruling out different cases corresponding to potential images of the gadget. This is achieved in two steps. First, we consider the subgraph $\Ical'$ of $\Hcal_n$ induced on $\{v_1,v_2,v_3,u_1,u_2,u_3,a\}$. Evidently, the map $\phi_n: \Ical' \to \Hcal_n$ sending 
\[ (v_1,v_2,v_3,u_1,u_2,u_3,a) \mapsto (v_{n-2},v_{n-1},v_n,u_{n-2},u_{n-1},u_n,b)\]
is an embedding. We argue that this is the only non-trivial embedding of $\Ical'$ in $\Hcal_n$.

\begin{lemma}\label{lem:embeds2}
    Let $n \geq 7$ and $f:\Ical' \to \Hcal_n$ be an embedding. Then $f$ is either the inclusion map or equal to $\phi_n$. 
\end{lemma}
 
 \begin{proof}
    As before, we write $V:=\{v_1,\dots,v_n\}\subseteq V(\Hcal_n)$ and $U:=\{u_1,\dots,u_n\}\subseteq V(\Hcal_n)$. We shall consider the possible images of the vertex $v_2$. Suppose that $f(v_2)=u_i$ for some $i \in [n]$. Clearly, since the vertices $v_1,v_3,a$ are pairwise non-adjacent, we cannot have $f[\{v_1,v_3,a\}] \subseteq U$. We hence distinguish cases. 
    \begin{enumerate}
        \item Suppose that $v_1,v_3,a$ are all mapped to vertices in $V$ under $f$. Since these are non-adjacent, we must have $f[\{v_1,v_3,a\}]=\{v_m,v_r,v_\ell\}$ for some $m + 2 < r + 1 < \ell \leq i$. Now, consider $f(u_1)$; this must be some vertex in $\Hcal_n$ which is adjacent to only one of $v_m,v_r,v_\ell$ and not adjacent to $u_i$. This necessarily implies that $f(u_1)=v_{i+1}$, $f(v_1)=v_\ell$ while $\ell = i$. Consider $f(u_2)$; this must be a vertex non-adjacent to $v_{i+1}$, and adjacent to $u_i,v_i$ and exactly one of $\{v_m,v_r\}$. From this we deduce that $f(u_2)=v_{i-1}$, $f(a)=v_r$, and $r=i-2$. Finally, the vertex $f(u_3)$ must be adjacent to $v_m,v_{i-2},v_i,u_i,v_{i+1}$ and non-adjacent to $v_{i-1}$; obviously no such vertex exists in $\Hcal_n$, and we thus obtain a contradiction. 
        \item Suppose that two of $v_1,v_3,a$ are mapped to vertices in $V$ and one is mapped to a vertex in $U$. In this case we must have that  $f[\{v_1,v_3,a\}]=\{u_m,v_r,v_\ell\}$ for some $m + 1 < r + 1 < \ell \leq i$. Consider $f(u_1)$; this must be non-adjacent to $v_i$ and adjacent to exactly one of $u_m,v_r,v_\ell$. This further results in two distinct cases. If $f(u_1)=v_{i+1}$, then we have $f(v_1)=v_\ell$ and $\ell=i$, which leads to a contradiction with an analogous argument to the above. If $f(u_1)=u_{i-1}$, then necessarily $f(v_1)=v_r$ while $m=i-2, r=i-1, \ell=i$. Considering $f(u_2)$, we now see that this vertex must be non-adjacent to $u_{i-1}$ and adjacent to $u_i,v_{i-1}$ and exactly one of $\{u_{i-2},v_i\}$; evidently there is no such vertex in $\Hcal_n$ and we thus obtain a contradiction. 
        \item Suppose that exactly one of $v_1,v_3,a$ is mapped to a vertex in $V$ and two are mapped to vertices of $U$. This forces that $f[\{v_1,v_3,a\}]=\{u_{m-1},u_{m},v_\ell\}$ for some $m < \ell \leq i$ and $m + 1 < i$. Again, consider $f(u_1)$; this is non-adjacent to $u_i$ and adjacent to exactly one of $\{u_m,u_{m+1},v_\ell\}$. Once again this leads to two options. If $f(u_1)=v_{i+1}$, then we have $f(v_1)=v_\ell$ and $\ell=i$, which leads to a contradiction as in Case $1$. On the other hand, if $f(u_1)=u_{i-1}$ then we necessarily obtain that $f(v_1)=u_{m-1}$ while $m=i-2$ and $\ell=i$. The vertex $f(u_2) \in \Hcal_n$ must then be non-adjacent to $u_{i-1}$ and adjacent to $u_{i-1},u_{i}$ and exactly one of $u_{i-2},v_i$; since there is no such vertex in $\Hcal_n$ we once again obtain a contradiction. 
        \item Suppose that one of $v_1,v_3,a$ is mapped to $a$ or $b$ under $f$. Since $u_3$ is adjacent to all of $v_1,v_2,v_3,a$ it must necessarily be that $f(u_3)=u_m$ for some $m > i+1$. The vertex $f(u_2)$ must then be non-adjacent to $u_m$, and adjacent to $u_i$ and exactly two of $f(v_1),f(v_3),f(a)$. As no such vertex exists in this case, we obtain a contradiction. 
\end{enumerate}
        Since the above cases lead to a contradiction, we see that $f(v_2) \notin U$. Since no $v_i$ for $i \in [n] \setminus \{2,n-1\}$ has three neighbours which induce an independent set, this necessarily implies that $f(v_2)$ is equal to $v_2$ or $v_{n-1}$. Assume that $f(v_2)=v_2$. Again, since $v_1,v_3,a$ share no edges, we must necessarily have $f[\{v_1,v_3,a\}]=\{v_1,v_3,a\}$. Since $u_3$ is adjacent to all of $v_1,v_2,v_3,a$ we see that $f(u_3)=u_m$ for some $m \geq 3$. As $u_2$ is non-adjacent to $u_3$ and adjacent to $v_2$ to exactly two of $v_1,v_3,a$, we see that $f(u_3)=u_3$ and $f(u_2)=u_2$, which in turn ensure that $f$ is the inclusion map. By similar reasoning, we deduce that if $f(v_2)$ is equal to $v_{n-1}$ then $f = \phi_n$ as required.  
\end{proof}

\begin{lemma}\label{lem:Iembeds}
Let $n\geq 7$ and $f:\Ical_n \to \Hcal_n$ be an embedding. Then $f$ is the inclusion map.   
\end{lemma}

\begin{proof}
    Let $f: \Ical_n \to \Hcal_n$ be an embedding. It follows by \Cref{lem:embeds2} that either $f$ is the inclusion map, or it is the map given by swapping the two induced copies of $\Ical'$, i.e. the map 
    \[ (v_1,v_2,v_3,u_1,u_2,u_3,a) \mapsto (v_{n-2},v_{n-1},v_n,u_{n-2},u_{n-1},u_n,b);\]
    \[ (v_{n-2},v_{n-1},v_n,u_{n-2},u_{n-1},u_n,b)\mapsto (v_1,v_2,v_3,u_1,u_2,u_3,a).\]
    Since $u_{n-2}$ is adjacent to $v_2$, the latter case would imply that $u_1$ is adjacent to $v_{n-1}$, which is a contradiction. Hence, $f$ is the inclusion map as claimed. 
\end{proof}

\begin{proposition}
    For each $n\geq 7$ the graphs $\Hcal_n$ are minimal induced models of $\phi$. 
\end{proposition}

\begin{proof}
    We fix some $n\geq 7$. We first argue that $\Hcal_n\models \phi$ for every $n \geq 7$. Indeed, we clearly have that 
    \[\Hcal_n \models I(v_1,v_2,v_3,v_{n-2},v_{n-1},v_n,u_1,u_2,u_3,u_{n-2},u_{n-1},u_n,a,b).\]
    Moreover, the set $U:=\{u_1,\dots,u_n\}\subseteq V(\Hcal_n)$ is precisely the set of neighbours of $v_1$ which are not $v_2$, while the set $V:=\{v_1,\dots,v_n\}\subseteq V(\Hcal_n)$ is precisely the set of non-neighbours of $v_1$ which are not $a$ or $b$. Evidently, we then have that for every vertex $v_i \in V\setminus \{v_1\}$ the vertex $v_{i-1} \in V$ is adjacent to $v_i$ and its neighbourhood over $U$ strictly contains that of $v_i$. Consequently $\Hcal_n \models \phi_2(\bar v,\bar u,a,b)$. Likewise, for every vertex $u_i \in U \setminus \{u_n\}$ the vertex $u_{i+1} \in U$ is non-adjacent to $u_i$ and its neighbourhood over $V$ strictly contains that of $u_i$. It follows that $\Hcal_n \models \psi_2(\bar v,\bar u,a,b)$, and so $\Hcal_n \models \phi$ as required. 
    
    Now, suppose that $H$ is a proper induced subgraph of $\Hcal_n$, and assume for a contradiction that $H \models \phi$, i.e. there are vertices $x_1,\dots,x_6,y_1,\dots,y_6,\alpha,\beta$ of $H$
    \[ H \models ( I(\bar x,\bar y,\alpha,\beta) \land 
[\phi_1(\bar x,\bar y,\alpha,\beta) \land \psi_1(\bar x,\bar y,\alpha,\beta) \to \phi_2(\bar x,\bar y,\alpha,\beta) \land \psi_2(\bar x,\bar y,\alpha,\beta) ]).\]
    Since these vertices induce a copy of $\Ical_n$, it follows by \Cref{lem:Iembeds} that 
    \[ (x_1,x_2,x_3,x_4,x_5,x_6,y_1,y_2,y_3,y_4,y_5,y_6,\alpha,\beta)=\]\[(v_1,v_2,v_3,v_{n-2},v_{n-1},v_n,u_1,u_2,u_3,u_{n-2},u_{n-1},u_n,a,b),\]
    and so $\Ical_n \leq H \lneq \Hcal_n$. Moreover, letting $U':= U \cap V(H)$ and $V':= V \cap V(H)$ we see that
    \begin{itemize}
        \item the elements in $V'$ have pairwise comparable neighbourhoods over $U'$, and the elements of $U'$ have pairwise comparable neighbourhoods over $V'$;
        \item the only neighbour of $v_n$ in $U'$ is $u_n$, and the only neighbour of $u_1$ in $V'$ is $v_1$;
        \item $u_n$ is adjacent to every element of $V$;
        \item if $x,y \in V'$ are adjacent and the $U'$-neighbours of $y$ are strictly more than the $U'$-neighbours of $x$ then there is some $i \in [n-1]$ such that $y=v_i$ and $x=v_{i+1}$, and there is a unique vertex in $U'$ that is adjacent to $y$ and not adjacent to $x$, namely $u_i$;
        \item if $x,y \in U'$ are adjacent and the $V'$-neighbours of $y$ are strictly more than the $V'$-neighbours of $x$ then there is some $i \in [n-1]$ such that $y=v_{i+1}$ and $y=v_{i}$, and there is a unique vertex in $V'$ that is adjacent to $y$ and not adjacent to $x$, namely $v_i$.
    \end{itemize}
    It follows that $H \models (\phi_1(\bar v,\bar u,a,b) \land \psi_1(\bar v,\bar u,a,b))$. Since $H \models \phi$ this implies that $H \models (\phi_2(\bar v,\bar u,a,b) \land \psi_2(\bar v,\bar u,a,b))$. However, since $V(H) \subsetneq V(\Hcal_n)$, there is some $i \in [4,n-3]$ such that $v_i \notin V(H)$ or $u_i \notin V(H)$. Assume the former, and let $i \in [4,n-3]$ be maximal such that $v_i \notin V(H)$. It follows that there is no vertex in $x \in V'$ that is adjacent to $v_{i+1}$ and its neighbourhood over $U'$ is strictly greater than that of $v_{i+1}$, contradicting that $H \models \psi_1(\bar v,\bar u,a,b)$. By a symmetric argument we obtain a contradiction if $u_i \notin V(H)$, and thus follows that $H \not\models \phi$. 
\end{proof}

\begin{theorem}
    Extension preservation fails on any hereditary graph class containing the graphs $\Hcal_n$ for arbitrarily large $n \in \N$. 
\end{theorem}

\begin{proof}
    Let $\C$ be a class of graphs containing the graphs $\Hcal_n$ for arbitrarily large $n$. Since the formula $\phi$ is preserved under extensions over the class of all finite graphs, it is in particular preserved under extensions over $\C$. Since $\C$ is hereditary and $\phi$ has infinitely many minimal induced models in $\C$, namely the graphs $\Hcal_n$, it follows by \Cref{lem:minimalmodels} that $\phi$ is not equivalent to an existential formula over $\C$. 
\end{proof}

Finally, we observe that the graphs $\Hcal_n$ have bounded (linear) cliquewidth, which is easily seen to be at most $4$. For this, we crucially use the fact that successive pairs are adjacent on one side and non-adjacent on the other. One could simplify the construction, e.g. by using adjacency to denote succession on both sides, but this would slightly increase the cliquewidth. 
    
\begin{corollary}
    Extension preservation fails on $\mathsf{Clique}_k$ for every $k \geq 4$.
\end{corollary}

As witnessed by the above, orders appear to provide strong counterexamples to extension preservation. In the next section we explore preservation in certain monadically stable classes, where no such issues are expected to arise.

\section{Extension preservation on strongly flip-flat classes}\label{sec:strflipflat}
Local information on dense graphs can be as complicated as global information, as for instance is the case with cliques. This fact seemingly renders locality unhelpful in the context of dense graph classes. Nonetheless, our understanding of tame classes indicates that it is still possible to recover meaningful local information, after possibly measuring distance in an alternative metric which comes from ``sparsifying'' our graphs in a controlled manner. The flip operation, which is central to the emerging theory of dense graph classes, plays precisely this role. We introduce it in the following definition.

\begin{definition}
    Let $G$ be a graph and $k \in \N$. A $k$-partition $P$ of $G$ is a partition of the vertex set into $k$ labelled parts $P_1,\dots,P_k$, i.e. $V(G)=\bigcup_{i \in [k]} P_i$ and $P_i \cap P_j = \emptyset$ for $i \neq j$. By a $k$-flip $F$ we simply a symmetric subset of $[k]^2$, i.e. a set of tuples $F=\{(i,j): i,j \in [k]\}$ such that $(i,j) \in F \iff (j,i) \in F$. 
    Given a $k$-partition $P$ of $G$ and a $k$-flip $F$ we define the graph $G \triangle_F P$ on the same vertex set as $G$ and on the edge set 
    \[ E(G \triangle_F P) := E(G) \triangle \{(u,v): u\neq v, u \in P_i, v \in P_j, \text{ and } (i,j) \in F\}.\]
    where $\triangle$ denotes the symmetric difference operation. 
\end{definition}

We note that the notation for flips existing in the literature uses the notation $\oplus$ rather than $\triangle$ (e.g. in \cite{indiscernibles}); here we have opted for the latter as the symbol $\oplus$ was used in \cite{extensionspreservation} and \cite{dawar2024preservation} to denote the amalgamation operation. Moreover, instead of partitioning our graph, we may define $k$-flips by applying a sequence of at most $k$ atomic operations, each one switching the edges and non-edges between two arbitrary subsets $A,B$ of our vertex set. Evidently, these definitions are equivalent up to blowing up the number of flips by a value that only depends on $k$, while we have opted for the partition definition here to simplify our construction in \Cref{def:flipsum} below.  

\begin{definition}\label{def:flipflat}
    We say that a hereditary class of graphs $\C$ is \emph{flip-flat}\footnote{The original definition of flip-flatness in \cite{indiscernibles} is the extension of our definition here to the hereditary closure of the class.} if for every $r \in \N$ there exist $k_r \in \N$ and a function $f_r: \N \to \N$ satisfying that for every $m \in \N$ and every $G \in \C$ of size at least $f_r(m)$ there is a $k_r$-partition $P$ of $G$, a $k_r$-flip $F \subseteq [k_r]^2$, and a set $A \subseteq V(G)$ of size at least $m$ which is $r$-independent in $G \triangle_F P$. If in the above $k_r:=k \in \N$ does not depend on $r$, then we say that $\C$ is \emph{strongly flip-flat}. 
\end{definition}

It was established in \cite[Theorem 1.3]{indiscernibles} that a hereditary class of graphs is flip-flat if, and only if, it is monadically stable. In particular, every transduction of a quasi-wide class is flip-flat. The qualitative difference between strong flip-flatness and flip-flatness is precisely the same as that of almost-wideness and quasi-wideness. 
We make this idea precise in the following straightforward proposition, which establishes that every transduction of a uniformly almost-wide class is strongly flip-flat. For this, we use the following lemma from \cite[Lemma H.3]{flipwidth}, which follows easily from Gaifman's locality theorem.

\begin{lemma}[Flip transfer lemma, \cite{flipwidth}]\label{lem:fliptransfer}
    There exists a (computable) function $\Xi : \N^3 \to \N$ satisfying the following. Fix $k,c,q \geq 1$ and $\Tcal_{\delta,\phi}$ a transduction involving $c$ colours and formulas of quantifier rank at most $q$. Let $G,H$ be graphs such that $H \in T_{\delta,\phi}(G)$. Then for every $k$-partition $P$ of $G$ and $k$-flip $F$ there exists a $\Xi(k,c,q)$-partition $P_H$ of $H$ and a $\Xi(k,c,q)$-flip $F_H$ such that for all $u,v \in V(H)$:
    \[ \dist_{G \triangle_F P}(u,v) \leq 2^q \cdot \dist_{H \triangle_{F_H} P_H}(u,v).\]
\end{lemma}

\begin{proposition}\label{prop:strongflip}
    Every transduction of a uniformly almost-wide graph class is strongly flip-flat. 
\end{proposition}

\begin{proof}
    Let $\C$ be a uniformly almost-wide graph class and fix $k_\C \in \N$ witnessing this, so that for every $r,m \in \N$ there is $f_r(m) \in \N$ satisfying that every $G$ of size at least $f_r(m)$ in the hereditary closure of $\C$ contains an $r$-independent set of size $m$ after removing at most $k_\C$ elements. 
    Let $\Dcal$ a class such that there is a transduction $T_{\delta,\phi}$ satisfying $\Dcal \subseteq T_{\delta,\phi}(\C)$. Let $c \in \N$ be the number of unary predicates used by $T$, and $q$ be the maximum of the quantifier ranks of $\delta$ and $\phi$. We argue that $\Dcal$ is strongly flip-flat with $k:=\Xi(2^{k_\C},c,q)$. 
    
    Indeed, fix $r,m \in \N$ and a graph $H \in \Dcal$ of size at least $f_{2^q\cdot r}(m)$. It follows that there exists some $G \in \C$ such that $H \in T_{\delta,\phi}(G)$, and since $f_{2^q\cdot r}(m)\leq |V(H)|$, we obtain by uniform almost-wideness that $G[V(H)]$ contains a $(2^q\cdot r)$-independent set of size $m$ after removing a set of size at most $k_\C$. In particular, there is a $2^{k_\C}$-partition $P$ of $G$ and a $2^{k_\C}$-flip $F$ such that $(G \triangle_F P)[V(H)]$ contains an $(2^q\cdot r)$-independent subset of size $m$; call this set $A$. Consequently, \Cref{lem:fliptransfer} implies that there is a $k$-partition $P$ of $H$ and a $k$-flip $F$ such that for all $a,b \in A \subseteq V(H)$
    \[ r = \frac{2^q\cdot r}{2^q} \leq \frac{\dist_{G \triangle_F P}(a,b)}{2^q} \leq \dist_{H \triangle_{F_H} P_H}(a,b),\]
    i.e. $A$ is an $r$-independent set of size $m$ in $H \triangle_{F_H} P_H$. It follows that $\Dcal$ is strongly flip-flat. 
\end{proof}

In particular, transductions of bounded degree classes, classes of bounded shrub-depth \cite{ganian2019shrub}, and transductions of proper minor-closed classes \cite[Theorem 5.3]{atseriaspreservation} are all strongly flip-flat. However, obtaining preservation via locality and wideness in the style of \cite{extensionspreservation, atseriaspreservation, dawar2010homomorphism, dawar2024preservation} 
additionally requires subtle closure assumptions. The proofs of the above articles are essentially structured into two parts. The first part argues via locality that for every formula $\phi$ preserved by extensions (or homomorphisms in the case of \cite{atseriaspreservation,dawar2010homomorphism,dawar2024preservation}) over a class $\C$ closed under substructures and disjoint unions there exist $r,m \in \N$ such that no minimal induced model of $\phi$ in $\C$ can contain an $r$-independent set of size $m$. In the second part, wideness is used to bound the size of minimal models of $\phi$, as large enough models would have to contain $r$-independent sets of size $m$, under the proviso that a bounded number of bottleneck points have been removed. To account for the removal of these points, we have to work with an adjusted formula $\phi'$ in an expanded vocabulary, together with suitably adjusted structures (\cite{ajtai} called these \emph{plebian companions}) on which we apply the argument of the first part. Working with $\phi'$, however, has translated the requirement of closure under disjoint unions to closure under a more involved operation which depends on the choice of bottlenecks. Consequently, preservation can fail on natural tame classes which do not satisfy this closure condition, e.g. for planar graphs \cite[Theorem 5.8]{dawar2024preservation}.  

In the context of vertex deletions, the corresponding operation was \emph{amalgamation over bottlenecks} \cite[Theorem 4.2]{dawar2024preservation}. Here, we must formulate a different operation to account for the fact that flips are required to witness wideness. This is precisely the construction below. 

\begin{definition}\label{def:flipsum}
    Given $k \in \N$, a graph $G$, a $k$-partition $P$ of $G$, and a $k$-flip $F\subseteq [k^2]$ we write $G\star_{(F,P)}G$ for the graph whose vertex set $V(G\star_{(F,P)}G):=V(G+G)$ is the same as the disjoint union of two copies of $G$, and whose edge set is 
    \[ E(G \star_{(F,P)}G):= E(G + G) \cup \{(u,v): u,v \text{ are in different copies of } G, u \in P_i, v \in P_j, (i,j) \in F\}.\]
    We call this the flip-sum of $G$ over $(F,P)$.   
\end{definition}

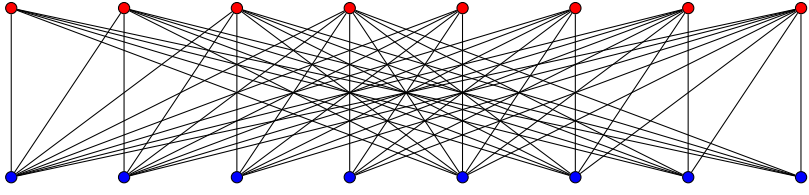
\begin{figure}[h!]
    \centering
    \begin{tikzpicture}[scale=1.5, every node/.style={circle, draw, inner sep=1.5pt}]
    \foreach \i in {1,2,...,4} {
        \node[fill=red] (u\i) at (\i, 1.5) {};
        \node[fill=blue] (v\i) at (\i, 0) {};
        \node[fill=red] (a\i) at (\i+4,1.5) {};
        \node[fill=blue] (b\i) at (\i+4,0) {};
    }

    \foreach \i in {1,...,4} {
        \foreach \j in {1,...,\i} {
            \draw (u\i) -- (v\j);
            \draw (a\i) -- (b\j);
        }
        \foreach \j in {1,...,4} {
            \draw (u\i) -- (b\j);
            \draw (v\i) -- (a\j);
        }
    }

\end{tikzpicture}
    \caption{The graph $H\star_{(F,P)}H$, where $H$ is the half-graph of order $4$, $P$ is the partition into red and blue vertices and $F=\{(1,2),(2,1)\}$.}
    \label{fig:enter-label}
\end{figure}

We now introduce the relevant translation for the formulas. 

\begin{definition}\label{def:phik}
    Given a $k$-flip $F\subseteq [k]^2$, consider the formula 
    \[ E_F(x,y):= E(x,y) \bigtriangleup_{(i,j) \in F} (P_i(x) \land P_j(y)). \]
    over the signature $\tau_{E}^k:=\tau_E\cup \{P_1,\dots,P_k\}$, where $\bigtriangleup_{(i,j) \in F}$ denotes the consecutive application of the $\mathsf{XOR}$ operator over all tuples $(i,j) \in F$. Given a $\tau_E$-formula $\phi$, we define the $\tau_E^k$-formula $\phi^k$ obtained from $\phi$ by replacing every atom $E(x,y)$ with the formula $E_F(x,y)$. Moreover, for every graph $G$ and $k$-partition $P$ we write $G_{(F,P)}$ for the $\{P_1,\dots,P_k\}$-expansion of $G \triangle_F P$ where each predicate is interpreted by the respective part of $P$. It is then clear from the definitions and the fact that the flip operation is involutive that
    \[ G \models \phi \iff G_{(F,P)} \models \phi^k.\]
\end{definition}

Our goal in \Cref{thm:main} is to start with a strongly flip-flat class and a formula $\phi$ and apply the argument of \cite[Theorem 4.3]{extensionspreservation} to the formula $\phi^k$ and the structures $G_{(F,P)}$. However, as previously explained, $\phi^k$ is not necessarily preserved under embeddings over $\C$. We can nonetheless use the following easy lemma in case that the class is closed under the desired flip-sums, which will be sufficient for our purposes.  

\begin{lemma}\label{lem:disjoint}
    Let $\C$ be a hereditary class of graphs and $\phi$ a formula preserved under extensions over $\C$. Fix a graph $G \in \C$, a $k$-partition $P$ of $G$, and a $k$-flip $F \subseteq [k]^2$. If $G\star_{(F,P)}G \in \C$ then 
    \[ G_{(F,P)} \models \phi^k \implies G_{(F,P)} + G_{(F,P)}[S] \models \phi^k\]
    for any $S \subseteq V(G)$. 
\end{lemma}

\begin{proof}
    Fix $\C, \phi, G, P, F$ as in the statement above, and let $S\subseteq V(G)$. Write $G^*$ for the subgraph of $G\star_{(F,P)}G$ induced on the vertex set of $G + G[S]$;  it follows that $G^* \in \C$ by hereditariness. As $G_{(F,P)} \models \phi^k$ we obtain that $G \models \phi$, and since $G^*$ contains an induced copy of $G$ and $\phi$ is preserved by extensions over $\C$ it follows that $G^* \models \phi$. Let $P^*$ be the natural $k$-partition of $G^*$ inherited from $G$, i.e. for each $i \in [k]$ the $i$-th part $P^*_i$ of $P^*$ contains the union of the $i$-th parts of $G$ and $G[S]$. It follows from the definitions that the structure $G^*_{(F,P^*)}$ is isomorphic to $G_{(F,P)} + G_{(F,P)}[S]$. Finally, since $G^* \models \phi$ we obtain that $G^*_{(F,P^*)}\models \phi^k$ and so $G_{(F,P)} + G_{(F,P)}[S] \models \phi^k$ as claimed. 
\end{proof}

We shall also make use of the following observation, which simply says that the induced substructures of $G_{(F,P)}$ are the same as expansions of flips of induced substructures of $G$. 

\begin{observation}\label{obs:indsub}
    Let $G$ be a graph, $P$ a $k$-partition of $G$, and $F$ a $k$-flip. Then for every $S \subseteq V(G)$ the structure $G_{(F,P)}[S]$ is equal to $G[S]_{(F,P_S)}$, where $P_S$ is the $k$-partition of $G[S]$ obtained by restricting each part of $P$ on $S$. 
\end{observation}

Before proceeding with \Cref{thm:main} we introduce some relevant definitions. Fix a relational signature $\tau$ and $q,d \in \N$, and let $A$ be a $\tau$-structure. By the \emph{$(q,d)$-type of some $a \in A$} we shall mean the set containing all the \MSO\ formulas $\theta(x)$ of quantifier rank\footnote{Here both first-order and second-order quantifiers contribute to the quantifier rank.} at most $q$, up to logical equivalence, such that $N_d^{A}(a)\models \theta(a)$. When we speak of a $(q,d)$-type $t$ over $\tau$, without reference to a particular element in a structure, we shall mean a $(q,d)$-type of some element in some $\tau$-structure. We say that an element $a \in A$ \emph{realises a $(q,d)$-type $t$} 
whenever $N_d^{A}(a) \models \theta(a)$ for all $\theta(x) \in t$. 
    Evidently, the number of $(q,d)$-types is bounded by some $p \in \N$ depending only on $\tau$ and $q$. Given a $\tau$-structure $A$, a set $C \subseteq A$, and a $(q,d)$-type $t$, we say that $t$ is \emph{covered by $C$ in $A$} if all $a \in A$ realising $t$ satisfy $N_d^{A}(a) \subseteq C$. For $n \in \N$ we also say that $t$ is \emph{$n$-free over $C$ in $A$} if there is a $2d$-independent set $S \subseteq A$ of size $n$ such that each $a \in S$ realises $t$ and $N_d^{A}(a)\cap C = \emptyset$.

\begin{lemma}\label{lem:covered}
    Fix a relational signature $\tau$ and $q,d \in \N$. Let $p$ be the number of $(q,d)$-types over $\tau$. Then for every $\tau$-structure $A$ and $n \in \N$, there exists a radius $e \leq 2dp$ and a set $D \subseteq A$ of at most $(n-1)p$ points such that each $(q,d)$-type is either covered by $N_e^A(D)$ or is $n$-free over over $N_e^A(D)$. 
\end{lemma}

\begin{proof}
    Fix an enumeration $t_1,\dots,t_p$ of all $(q,d)$-types over $\tau$. We shall define $D$ and $e$ inductively starting at $D_0=\emptyset$ and $e_0=0$. Assuming $D_i$ and $e_i$ have been defined, we let $C = N_{e_i}^{G^*}(D_i)$. If all types are covered by $C$ or are $n$-free over $C$ then we are done; otherwise, we let $j \in [p]$ be minimal such that $t_j$ is neither covered by $C$ nor $n$-free over $C$. We then define a set $E \subseteq A$ inductively, starting with $E_0:=\emptyset$ and at step $\ell+1$ adding to $E_\ell$ a realisation $a \in A\setminus N_{2d}^A(C \cup E_\ell)$ of $t_j$ if there exists one; this iteration must stop within $n-1$ steps, as otherwise $t_j$ would be $n$-free over $C$. In particular, $|E|\leq n-1$ and $t_j$ is covered by $N_{e_i+2d}^A(D_i \cup E)$. We subsequently let $D_{i+1}=D_i \cup E$ and $e_{i+1}=e_i + 2d$. It follows that the construction must stop within at most $p$ steps, since at each step we cover a previously uncovered type, which in addition, remains covered for the rest of the construction. Consequently, $|D|\leq (n-1)p$ and $e \leq 2dp$ as claimed. 
\end{proof}

\begin{restatable}{theorem}{ThmMain}\label{thm:main}
    Fix a hereditary class of graphs $\C$. Suppose that there is some $k \in \N$ such that for all $r \in \N$ there is a function $f_r:\N \to \N$ satisfying that for every $m \in \N$ and every $G \in \C$ of size at least $f(m)$ there is a $k$-partition $P$ of $V(G)$, some $k$-flip $F $, and $A \subseteq V(G)$ such that 
    \begin{enumerate}
        \item $|A|\geq m$;
        \item $A$ is $r$-independent in $G \triangle_F P$;
        \item\label{ass:3} $G\star_{(F,P)}G \in \C$.
    \end{enumerate}
    Then extension preservation holds over $\C$.
\end{restatable}

The proof of \Cref{thm:main} is an adaptation of the proof of \cite[Theorem 4.3]{extensionspreservation}, which established that extension preservation holds over any class closed under weak substructures and disjoint unions which is \emph{wide}, i.e. for every $r \in \N$ there exists $f_r:\N \to \N$ such that for every $m \in \N$ every structure with at least $f_r(m)$-many elements contains an $r$-independent set of size $m$. Here, we replace wideness by strong flip-flatness by working with the formula $\phi^k$ of \Cref{def:phik}.  Moreover, as previously explained, addability is replaced by closure under the appropriate flip-sums; preservation is then ensured by \Cref{lem:embeds2}. Finally, we may relax the assumption of closure under weak substructures to closure under induced substructures.

\begin{proof}[Proof of \Cref{thm:main}]
    Fix $\C$ as above, and let $\phi$ be a formula preserved by embeddings overs $\C$. We shall obtain a bound on the size of the minimal induced models of $\phi$, by arguing that any large enough model of $\phi$ contains a proper induced substructure which also models $\phi$. We can then conclude that $\phi$ is equivalent to an existential formula over $\C$ using \Cref{lem:minimalmodels}. 

    Letting $k \in \N$ be as the in the statement of \Cref{thm:main}, we consider the formula $\phi^k$ from \Cref{def:phik}. Using Gaifman's locality theorem we rewrite $\phi^k$ into a boolean combination of basic local sentences, i.e. we may assume that there is some $\ell \in \N$ and $\tau_E^k$-sentences $\psi_i$ for $i \in [\ell]$ such that  
    \[ \phi^k = \bigvee_{i \in \ell} \psi_i \text{ and }\psi_i = \bigwedge_{j \in A_i} \chi_{ij} \land \bigwedge_{j \in B_i} \neg \chi_{ij},\]
    where each $\chi_{ij}$ is a basic local sentence. We henceforth fix the following constants:
    \begin{itemize}
        \item $\rho$ is the maximum over all the locality radii of the $\chi_{ij}$; 
        \item $s$ is the sum of all widths of the $\chi_{ij}$;
        \item $\gamma$ is the maximum over all the quantifier ranks of the $\chi_{ij}$;
        \item $q: = \gamma+3\rho+3$;
        \item $d:= 2(\rho+1)(\ell +1)s + 6\rho +2$; 
        \item $p$ is the number of $(q,d)$-types over the signature $\tau_E^k$;
        \item $n:=(\ell +2)s$;
        \item $m:=(n-1)q + s + \ell s + 1$;
        \item $r:=4dp + 2\rho + 1$.
    \end{itemize}
    Our goal is to establish that any minimal induced model of $\phi$ in $\C$ must have size less than $f_r(m)$, where $f$ is as in the statement of \Cref{thm:main}. So, assume that some $G\models \phi$ has size at least $f_r(m)$. It follows by assumption that there is a $k$-partition $P$ and a $k$-flip $F$ such that $G \triangle_F P$ contains an $r$-independent set of size $m$. We henceforth work with the structure $G^*:=G_{(F,P)}$, i.e. the expansion of $G\triangle_F P$ with unary predicates corresponding to the parts of $P$. By definition, we have that $G^*\models \phi^k$. 

    By \Cref{lem:covered} we obtain a radius $e \leq 2dp$ and a set $D \subseteq V(G^*)$ of at most $(n-1)p$ vertices such that each $(q,d)$-type in $G^*$ is either covered $N_e^{G^*}(D)$ or is $n$-free over $N_e^{G^*}(D)$; we henceforth refer to types of the former kind as \emph{rare}, and to types of the latter kind as \emph{frequent}. 

    We proceed to inductively construct increasing sequences of sets $S_0 \subseteq S_1 \subseteq \dots \subseteq V(G^*)$, $C_0 \subseteq C_1 \subseteq \dots \subseteq V(G^*)$, and $I_0 \subseteq I_1 \subseteq \dots \subseteq I$ which satisfy the following conditions for every $i$:
    \begin{enumerate}
        \item\label{it:cond1} $S_i \subseteq N_\rho^{G^*}(C_i)$;
        \item $|C_i| \leq is$;
        \item $|I_i| = i$;
        \item\label{it:ext} no disjoint extension of $G^*[S_i]$ satisfies $\bigvee_{j \in I_i} \psi_j$;
        \item\label{it:frqnt} $N_e^{G^*}(D)$ and $N_d^{G^*}(C_i)$ are disjoint.
    \end{enumerate}
    
    Clearly, this construction must terminate within $\ell$ steps. Indeed, assume for a contradiction that we have constructed $S_\ell,C_\ell,$ and $I_\ell$ satisfying conditions \ref{it:cond1}-\ref{it:frqnt} above. If so, then $I_\ell = I$ while $G^* + G^*[S_\ell]$ is a disjoint extension of $G^*[S_\ell]$ which satisfies $\phi^k=\bigvee_{i \in I} \psi_i$ by \Cref{lem:disjoint}, therefore contradicting condition \ref{it:ext}. At the end of the construction we will obtain some $N < \ell$ and some $S_N \subsetneq V(G^*)$ satisfying $G^*[S_N] \models \phi^k$. Combining \Cref{def:phik} with \Cref{obs:indsub}, this will imply that $G[S_N] \models \phi$, and hence that $G$ cannot be a minimal model of $\phi$ as required. 
    
    Initially, we set $S_0=C_0=I_0 = \emptyset$. Assume that $S_i,C_i,$ and $I_i$ have been defined. Write $H^*:= G^* + G^*[S_i]$ for the disjoint union of $G^*$ with its substructure induced on $S_i$. By our closure assumptions on $\C$ and \Cref{lem:disjoint} we deduce that $H^* \models \phi^k$. In particular, there exists some $i' \in I$ such that $H^* \models \psi_{i'}$, while $i' \notin I_i$ due to property \ref{it:ext}. We let $I_{i+1}=I_i \cup \{i'\}$ and henceforth drop the reference to the index $i'$ as it will remain fixed for the remaining of the argument, e.g. by writing $\psi$ and $\chi_j$ instead of $\psi_{i'}$ and $\chi_{i'j}$ respectively. 
    
    As $H^*$ satisfies $\psi= (\bigwedge_{j \in A} \chi_j \land \bigwedge_{j \in B} \neg \chi_j)$, it satisfies the basic local sentences $\chi_j$ with $j \in A$. 
    For each $j \in A$, we may thus choose a minimal set $W_j \subseteq V(H^*)$ of witnesses for the outermost existential quantifiers of the basic local sentence $\chi_j$, and let $W:=\bigcup_{j \in A} W_j$ be their union. As $s$ is the sum of the widths of all the $\chi$'s it follows that $|W|\leq s$. We partition $W$ into those witnesses that appear in the disjoint copy of $G^*$, and those that appear in the disjoint copy of $G^*[S_i]$, and write $W_G$ and $W_H$ for these respective parts. 
    
    Now, suppose that some $v \in W_G$ satisfies $N_{\rho+1}^{G^*}(C_i) \cap N_{\rho}^{G^*}(v)\neq \emptyset$; we argue that we may replace $v$ with some witness $v' \in V(G^*)$ such that $N_{\rho+1}^{G^*}(C_i) \cap N_{\rho}^{G^*}(v')=\emptyset$. Indeed, we first choose some $u \in C_i$ such that $N_{\rho+1}^{G^*}(u) \cap N_{\rho}^{G^*}(v)\neq \emptyset$. Consequently, we have that $N_{\rho}^{G^*}(v)\subseteq N_{3\rho+1}^{G^*}(u)\subseteq N_d^{G^*}(u)$. Property \ref{it:frqnt} then ensures that the $(q,d)$-type $t$ (in $G^*$) of $u$ is frequent, and so it has $n > (\ell+1)s \geq |W\cup C_i|$ realisations whose $d$-neighbourhoods are pairwise disjoint and disjoint from $N_e^{G^*}(D)$. We may thus pick a realisation $u' \in V(G^*)$ of $t$ such that $N_{\rho+1}^{G^*}(W \cup C_i)\cap N_{3\rho+1}^{G^*}(u')=\emptyset$. Let $\tau$ be the $(\gamma,\rho)$-type of $v$, and consider the formula 
    \[ \theta(x):=\exists y[ \forall z (\dist(y,z)\leq \rho \to \dist(x,z) \leq 3\rho+1) \land \bigwedge_{\eta \in \tau} \eta^{N_r(y)}(y) ]. \]
    Clearly, the quantifier rank of $\theta$ is bounded by $3\rho+3+\gamma\leq q$, while $N_d^{G^*}(u) \models \theta(u)$ with $v$ serving as the existential witness. Consequently $\theta(x)$ is in $t$, and as $u$ and $u'$ have the same $(q,d)$-type, it follows that $N_d^{G^*}(u') \models \theta(u')$. It follows that there is $v' \in V(G^*)$ such that $N_{\rho}^{G^*}(v')\subseteq N_{3\rho+1}^{G^*}(u')\subseteq N_d^{G^*}(u')$, while $v$ and $v'$ have the same $(\gamma,\rho)$-type. In particular, their $\rho$-neighbourhoods satisfy the same \FO-formulas of quantifier rank $\leq \gamma$. Finally, observe that $N_{\rho+1}^{G^*}(W \cup C_i)\cap N_{3\rho+1}^{G^*}(v')=\emptyset$ and so $N_{\rho+1}^{G^*}(C_i)\cap N_{\rho}^{G^*}(v') = \emptyset$; we may thus replace $v$ by $v'$ in $W_G$ as a witness. 

    After replacing all such witnesses in $_G$, we can ensure that 
    \[|\{v \in W_G: N_{\rho+1}^{G^*}(C_i) \cap N_{\rho}^{G^*}(v)\neq \emptyset\}| = 0 \tag{$\star$}\]
    Consider the induced substructure $U^*:=G^*[N_e^{G^*}(D)\cup N_\rho^{G^*}(W_G) \cup S_i]$. We claim that $U^*$ satisfies $\bigwedge_{j \in A} \chi_j$. Indeed, notice that $S_i\subseteq N_{\rho}^{G^*}(C_i)$, while $N_{\rho+1}^{G^*}(C_i)$ is disjoint from $N_e^{G^*}(D)$ by property \ref{it:frqnt} and disjoint from $N_\rho^{G^*}(W_G)$ by $(\star)$. It follows that $U^*$ is the disjoint union of $G^*[N_e^{G^*}(D)\cup N_\rho^{G^*}(W_G)]$ and $G^*[S_i]$; thus all the witnesses from $W$ and their $\rho$-neighbourhoods can be found in $U^*$, implying that $U^*\models \chi_j$ for all $j \in A$ as these are basic local sentences. 

    Now, observe that $U^*$ is a proper induced substructure of $G^*$. This is because 
    \[ |D\cup W_G \cup C_i|\leq (n-1)p+s+\ell s < m; \]
    \[ N_e^{G^*}(D)\cup N_\rho^{G^*}(W_G) \cup S_i \subseteq N_{2dp+\rho}^{G^*}(D \cup W_G \cup C_i) \subseteq N_{\left \lfloor{r/2}\right \rfloor}^{G^*}(D \cup W_G \cup C_i), \]
    and so, unlike $G^*$, $U^*$ does not contain an $r$-independent set of size $m$. Consequently, if $U^* \models \phi^k$ then we set $S_N:=N_e^{G^*}(D)\cup N_\rho^{G^*}(W_G) \cup S_i$ and our construction terminates.
    
    We hereafter assume that $U^*\not\models\phi^k$, and proceed with the definition of $S_{i+1}$ and $C_{i+1}$. Since $U^* \models \bigwedge_{j \in A}\chi_j$ it must be that $U^*\not\models \bigwedge_{j \in B}\neg \chi_j$. We can therefore fix some $j \in B$ such that $U^* \models \chi_j$. 
    Suppose that 
    \[ \chi_j = \exists x_1 ,\dots, \exists x_{s'} [ \bigwedge_{a\neq b }\dist(x_a,x_b)>2\rho' \land \bigwedge_{a} \xi^{N_{\rho'}(x_a)}(x_a)]\]
    for some $\rho'\leq \rho$, $s'\leq s$, and a formula $\xi$ of quantifier rank $\gamma'\leq \gamma$. Fix a set $V=\{w_1,\dots,w_{s'}\}\subseteq U^*$ of witnesses for the outermost existential quantifier of $\chi_j$. 
    Notice that if the $(q,d)$-type in $G^*$ of every $w \in V$ was rare then $N_{\rho'}^{G^*}(V)\subseteq N_e^{G^*}(D) \subseteq V(G^*)$, implying that $G^*\models \chi_j$ and thus $H^* \models \chi_j$ as $H^*$ is a disjoint extension of $G^*$ and $\chi_j$ is a basic local sentence. We can thus fix some $w \in V$ whose $(q,d)$-type in $G^*$, say $t_w$, is frequent. As a result, there is a set $Z\subseteq V(G^*)$ of $n$ realisations of $t_w$ whose $d$-neighbourhoods are pairwise disjoint and disjoint from $N_e^{G^*}(D)$. Now, since $4\rho+ 3 \leq d$, $n=(\ell+2)s$, and $|C_i|\leq \ell s$, there exists a subset $Z'\subseteq Z$ of at least $s$ elements which additionally satisfies $N_{\rho+1}^{G^*}(C_i)\cap N_{\rho}^{G^*}(Z') = \emptyset$. 

    Consider $F:=N_{\rho'}^{U^*}(w)$. Evidently, $U^*[F]=G^*[F]$ and so $G^*[F] \models \xi^{N_{\rho'}(x)}(w)$. For a set variable $X$ consider the formula $\xi^{N_{\rho'}(x)\cap X}(x,X)$ obtained from $\xi$ by simultaneously relativising the quantifiers of $\xi$ to the $r'$-neighbourhoods of $x$ and to the set $X$. Observe that the quantifier rank of $\xi^{N_{\rho'}(x)\cap X}(x,X)$ is at most $\gamma'+\rho' < q$, and moreover $G^* \models \xi^{N_{\rho'}(x)\cap X}(w,F)$. Since $\rho'<d$ it follows that the \MSO\ formula $\exists X \xi^{N_{\rho'}(x)\cap X}(x,X)$ is in $t_w$. As every $\omega \in Z'$ has the same $(q,d)$-type in $G^*$ as $w$, we may find sets $F_\omega \subseteq N_{\rho'}^{G^*}(\omega)$ for every $\omega \in Z'$ such that $G^* \models \xi^{N_{\rho'}(x)\cap X}(\omega,F_\omega)$. In particular, this implies that $G^*[F_\omega] \models \xi^{N_{\rho'}(x)}(\omega)$. We finally let: 
    \[ C_{i+1}=C_i \cup Z'; \quad  S_{i+1}=S_i \cup \bigcup_{\omega \in Z'} F_\omega.\]
    We argue that these satisfy the properties \ref{it:cond1}-\ref{it:frqnt}. First, observe that $|C_{i+1}|=|C_i|+s \leq is + s = (i+1)s$. Moreover, as $F_\omega \subseteq N_{\rho'}^{G^*}(\omega)$, $\omega \in C_{i+1}$, and $\rho' \leq \rho$ we have $S_{i+1} \subseteq N_{\rho}^{G^*}(C_{i+1})$. By the fact that every $\omega \in Z'$ realises a frequent type we also have that $N_e^{G^*}(D)\cap N_d^{G^*}(C_{i+1})=\emptyset$. It remains to argue that no disjoint extension of $G^*[S_{i+1}]$ satisfies $\bigvee_{j \in I_{i+1}} \psi_j$. 

    Towards this, we note that $G^*[S_{i+1}]$ is a disjoint extension of $G^*[S_i]$ by the fact that $S_i \subseteq N_{\rho}^{G^*}(C_i)$ and $N_{\rho+1}^{G^*}(C_i)\cap N_{\rho}^{G^*}(Z') = \emptyset$. Therefore, no disjoint extension of $G^*[S_{i+1}]$ satisfies $\psi_j$ for $j \in I_i$. At the same time, every disjoint extension of $G^*[S_{i+1}]$ contains witnesses for the outermost existential quantifiers of $\chi_{i'j}$, namely the elements $\omega \in Z'$, which are pairwise at distance at least $2d>2\rho'$ and satisfy $G^*[F_\omega]\models \xi^{N_{\rho'}(x)}(\omega)$ and $N_{\rho'}^{G^*[S_{i+1}]}(\omega)=F_\omega$, and thus $G^*[S_{i+1}]\models \xi^{N_{\rho'}(x)}(\omega)$. It follows that every disjoint extension of $G^*[S_{i+1}]$ satisfies $\chi_{i'j}$, and so it cannot satisfy $\psi_{i'}$ as needed. This complete our inductive construction of $S_{i+1},C_{i+1}$, and $I_{i+1}$. 
    \end{proof}

    \begin{example}
    For $d \in \N$, write $\Dcal_d$ be the class of all graphs $G$ such that the maximum degree of $G$ is at most $d$, or the maximum degree of $\overline G$, i.e. the complement graph of $G$, is at most $d$. Let $f_r(m)=(m-1)(d+1)^r+1$ and consider a graph $G \in \Dcal_d$ of size at least $f_r(m)$. If $G$ has maximum degree $d$, then $G$ must contain an $r$-independent set of size $m$. Consequently, letting $P=\{V(G)\}$ and $F=\emptyset$, we see that $G \star_{(F,P)} G$ is simply the disjoint union of two copies of $G$, which still has maximum degree $d$ and is therefore in $\Dcal_d$. On the other hand if $\overline G$ has maximum degree $d$, then for $P=\{V(G)\}$ and $F=\{(1,1)\}$, we see that $\bar G=G \triangle_F P$ has an $r$-independent set of size $m$. Since $G \star_{(F,P)} G$ is the complement of the disjoint union of two copies of $\overline G$ and so $\overline{G \star_{(F,P)} G}$ has maximum degree $d$, it follows that ${G \star_{(F,P)} G}\in \Dcal_d$. Consequently, extension preservation holds over $\Dcal_d$ by \Cref{thm:main}. 
\end{example}

As mentioned above, \Cref{lem:minimalmodels} implies that any well-quasi-ordered class has the extension preservation property. In particular, this applies to classes of bounded \emph{shrubdepth} \cite[Corollary 3.9]{ganian2019shrub}. Still, in the following example we indirectly show that the class of all graphs of \emph{$\SC$-depth} at most $k$ has extension preservation by showing that it satisfies the requirements of \Cref{thm:main}. This is an illustration that, although closure under flip-sums is a technical condition, it can be present in interesting tame dense classes.

\begin{definition}\cite[Definition 3.5]{ganian2019shrub}
    We inductively define the class $\SC(k)$ as:
    \begin{itemize}
        \item $\SC(0)=\{K_1\}$;
        \item If $G_1,\dots,G_n \in \SC(k), H:=G_1 + \dots + G_n$ and $X \subseteq V(H)$, then $\overline H^X:=H \triangle_F P \in \SC(k+1)$ for $P_1:= X, P_2:=V(H)\setminus X$ and $F=\{(1,1)\}$, i.e. $\overline H^X$ is the graph obtained from $H$ by flipping the edges within $X$. 
    \end{itemize}
\end{definition}

\begin{example}
    Fix $k \in \N$ and let $G \in \SC(k)$. Consider an \emph{$\SC$-decomposition tree of $G$}, i.e. a labelled tree $\Tcal$ of height $k+1$ whose leaves are labelled by the vertices of $G$, every non-leaf node is labelled by the graph $\overline{(G_1 + \dots G_n)}^X$ where $G_i$ are the labels of its children, $X$ is a subset of $\bigcup_{i \in [n]} V(G_i)$, and the root $\rho$ is labelled by $G$. Let $f(m)=m^{k+1}$, and suppose that $|G|> f(m)$. Since $\Tcal$ has height $k+1$ and its leaves correspond to the vertices of $G$, there must exist some vertex $t$ of $\Tcal$ with at least $m$ children. Let $t_1:=\rho,t_2,\dots,t_\ell:=t$ be the unique path from the root of $\Tcal$ to $t$, and for each $i \in [\ell]$ let $X_i\subseteq V(G)$ be the set coming from the label of $t_i$. Letting $P$ be the partition of $V(G)$ into $2^\ell$ parts depending on the membership of a vertex within each of $X_1,\dots,X_\ell$ and $F\subseteq [2^\ell]^2$ be the flip that corresponds to complementing each of $X_1,\dots,X_\ell$, it follows that $G \triangle_F P$ contains at least $m$ distinct connected components. 
    Let $\Tcal'$ be the tree obtained from $\Tcal$ by the following operation. We first create a copy of each subtree of $\Tcal$ rooted at a child of $t_\ell$ and connect them to $t_\ell$. The labels are naturally carried from each original subtree to the copy. If the label of $t_\ell$ in $\Tcal$ was $\overline{(G_1+\dots+G_m)}^X$, then its label in $\Tcal'$ is $\overline{(G_1+G'_1+\dots+G_m+G_m')}^{X\cup X'}$ where each $G_i'$ corresponds to the copy of $G_i$, and $X'$ corresponds to the set of copies of the vertices in $X$. From there, we perform the same operation for $i=\ell-1,\dots,1$, this time  copying only the children of $t_i$ that are not $t_{i+1}$. This completes the construction of $\Tcal'$. It is easy to then see than the root of $\Tcal'$ corresponds to the graph $G \star_{(F,P)} G$, thus witnessing that $G \star_{(F,P)} G \in \SC(k)$. It follows that the class $\SC(k)$ satisfies the requirements of \Cref{thm:main}, and thus extension preservation holds over this class. 
\end{example}

\section{Conclusion}

We conclude with some questions and remarks. Firstly, it would be of independent interest to provide a characterisation of strongly flip-flat classes, akin to the characterisation of almost-wide  classes via shallow minors given in \cite[Theorem 3.21]{firstorderproperties}. This could either be a characterisation via excluded induced subgraphs occurring in flips, in analogy to the one of monadic stability provided in \cite{dreier2023first}, or in terms of \emph{shallow vertex minors}, in analogy to the one in \cite{shallowvertexminors}. 

Moreover, it is unclear whether one can produce a formula preserved by extensions with minimal induced models of cliquewidth $3$. The issue with  interweaving definable orders is that one simultaneously requires for two sets to semi-induce a half-graph while (non-)adjacency is used to mark successors; this requires to keep track of at least four colours classes in a clique decomposition. 
We therefore leave the question of whether extension preservation holds over graphs of cliquewidth $3$ open. It is also easily seen that the structures $\Hcal_n$ have \emph{twin-width} $2$ (see \cite{bonnet2021twin} for definitions).  The status of extension preservation on the class of all graphs of twin-width $1$ is also unknown. 

The role of orders was crucial in our construction in \Cref{sec:ext}. In the context of undirected graphs, orders are instantiated through half-graphs. It natural to then inquire if, for every fixed $k,\ell \in \N$, the class of all graphs of cliquewidth at most $k$ which omit semi-induced half-graphs of size larger than $\ell$ has the extension preservation property. Every such class is known to be equal to a transduction of a class of bounded treewidth by \cite{rankwidthmeets}, and so by \Cref{prop:strongflip}, it is strongly flip-flat. It would therefore be interesting to provide a direct combinatorial argument witnessing this, so as to be able to verify if such classes satisfy the closure requirements of \Cref{thm:main}.

\bibliographystyle{plain}
\bibliography{bibliography}

\end{document}